%% file: IC_Classification.tex
\documentclass[11pt]{article}

\usepackage[margin=1in]{geometry}  
\usepackage{graphicx}              
\usepackage{amsmath}               
\usepackage{amsfonts}
\usepackage{dsfont}
\usepackage{amssymb}       
\usepackage{amsthm}                
\usepackage[inline]{enumitem}
\usepackage{color, soul}	
 \usepackage{booktabs}
 \usepackage[table,xcdraw]{xcolor}
\usepackage{hyperref}
\usepackage[titletoc,title]{appendix}
\usepackage{algpseudocode}
\hypersetup{
	colorlinks=true,
	urlcolor=blue,
	linkcolor=blue,
	citecolor=red
	}

\newcolumntype{L}{>{\arraybackslash}m{3cm}}

\usepackage{tikz}
\usepackage{pgfplots}

\pgfplotsset{compat=1.10}
\usepgfplotslibrary{fillbetween}
\usetikzlibrary{patterns}

\newtheorem{thm}{Theorem}
\newtheorem{lem}[thm]{Lemma}
\newtheorem{prop}[thm]{Proposition}

\newtheorem{defn}[thm]{Definition}

\newtheorem{clm}[thm]{Claim}

\newcommand{\cl}[1]{\mathcal{#1}} 
\newcommand{\NN}{\mathbb{N}}      

\newcommand{\citet}[1]{\cite{#1}}

\makeatletter
\renewcommand*{\@fnsymbol}[1]{\ensuremath{\ifcase#1\or *\or 
		\mathsection\or \mathparagraph\or \|\or **\or \dagger\dagger
		\or \ddagger\ddagger \else\@ctrerr\fi}}
\makeatother

\title{Incentive-Compatible Classification\thanks{This project has received funding from the European Research Council (ERC) under the European Union's Horizon 2020 research and innovation programme (grant agreement  n$^o  $ 740435).}}
\author{Yakov Babichenko \quad\qquad and \quad\qquad Oren Dean \qquad\quad  and \quad\qquad Moshe Tennenholtz\\
	\small{
		yakovbab@tx.technion.ac.il \qquad\qquad orendean@campus.technion.ac.il \qquad\quad\quad moshet@ie.technion.ac.il\quad\qquad
	}\\
	Technion --- Israel Institute of Technology\\	Haifa, Israel
}
\begin{document}
\maketitle
\input{classification-body}

\bibliographystyle{siam}
\bibliography{paper}

\input{classification-appx}
\end{document}

%% file: classification-body.tex
\begin{abstract}
We investigate the possibility of an incentive-compatible (IC, a.k.a. strategy-proof) mechanism for the classification of agents in a network according to their reviews of each other. In the $ \alpha $-classification problem we are interested in selecting the top $ \alpha $ fraction of users. We give upper bounds (impossibilities) and lower bounds (mechanisms) on the worst-case coincidence between the classification of an IC mechanism and the ideal $ \alpha $-classification.\\
We prove bounds which depend on $ \alpha $ and on the maximal number of reviews given by a single agent, $ \Delta $. Our results show that it is harder to find a good mechanism when $ \alpha $ is smaller and $ \Delta $ is larger. In particular, if $ \Delta $ is unbounded, then the best mechanism is trivial (that is, it does not take into account the reviews). On the other hand, when $ \Delta $ is sublinear in the number of agents, we give a simple, natural mechanism, with a coincidence ratio of $ \alpha $. 
\end{abstract}

\section{Introduction}
There are many situations in which peer agents have binary, directed interactions with each other, and in which one side can rate his experience from the interaction, or rather rate the other agent. The following are just a few examples:
\begin{enumerate}
	\item E-commerce sites in which buyers might also be sellers (e.g., ebay.com, amazon.com).
	\item Academic paper reviewers for a conference might themselves be authors of papers submitted to the same conference.
	\item Employees in an organisation are sometime asked to fill out a sociometric overview of their fellow friends.
\end{enumerate}
In all of the above examples, a coordinator/manager is classifying the agents in the system according to the reviews they received. In the e-commerce example, the top-rated sellers will appear higher and more frequently in search results; the academic conference will only accept a certain top-rated portion of the papers. Employees with higher sociometric results have better chances at a promotion. A natural problem arises---in order to maximize one's relative rating, it is a dominant strategy in these situations to give a harsh critique to all interactions. In this paper, we model the agents and their interactions as a directed network and ask whether it is possible to offer an incentive-compatible (IC) mechanism to select a subset which represents the top-rated (``worthy'') agents. We measure the quality of a mechanism as the resemblance between the selected set of the mechanism and the set of top-rated agents, in the worst case. We investigate the relation between the quality of the best possible mechanism to two parameters: \begin{enumerate*}[label=(\roman*)]
	\item the maximal number of reviews a single agent can issue (the maximal out-degree in the network, denoted $ \Delta $), and 
	\item the fastidiousness of the system (the relative size of the selected set, denoted $ \alpha $).
\end{enumerate*}   

\subsection{Our contribution}
In this paper we investigate the existence of an $ \alpha $-classification IC mechanism in weighted networks. Weighted networks without any limitation were not considered before as a framework for selection mechanisms (\cite{Kurokawa:2015:IPR:2832249.2832330} considered only $ \Delta $-regular weighted networks; the assumption of $ \Delta $-regularity somewhat simplifies the optimisation criterion since in this case the optimisation for average in-weights is equivalent to the optimisation for the sum of in-weights.). The most significant novelty of our model is the consideration of mechanisms which classify the agents to ``worthy'' and ``unworthy'', in contrast to previous works which only considered $ k $-subset selection ($ k $-selection) mechanisms. Our optimisation criterion is the coincidence between the mechanism's classification and the ideal classification. The difference from $ k $-selection is two-fold. First, we do not know the exact size of the set which we need to select. Second, we try to select as many of the {right} (truly worthy) agents and not the {wrong} agents regardless of how high their in-degree is; this is very different from $ k $-selection, which just looks for a subset with high in-degree (even if it is completely disjoint to the optimal subset).

We prove upper and lower bounds on the quality of the best possible IC classification mechanism. We chart the behaviour of these bounds as a function of $ \Delta $ and $ \alpha $. We show that as $ \Delta $ grows (agents are allowed to review many others), the possibility of an IC classification mechanism narrows down until for large values of $ \Delta $ the best mechanism is one of the two trivial mechanisms: select all the agents as worthy, or select every agent independently with probability 1/2. We show the reverse behaviour with $ \alpha $: as we lower $ \alpha  $ (the system is more picky about its worthy-classified agents), the quality of the best possible IC classification mechanism decays to zero. \\
On the other hand, for fixed $ \alpha $ and for $ \Delta $ which is negligible with respect to the number of agents, we provide a mechanism with a positive quality. The idea behind this mechanism is based on a well-known practice to partition the agents into three subsets: absolutely worthy, borderline, and absolutely unworthy. Unlike the well-known practice, our mechanism suggests to classify an agent into these categories after ignoring his reviews on others. This makes the mechanism IC, but complicates the performance analysis of the mechanism.\\
Previous works (e.g. \cite{AFPT11}, \cite{Kurokawa:2015:IPR:2832249.2832330}) showed the existence of an optimal $ k $-selection mechanisms when $ k $ is large (say $ k=\omega(1) $ or $ k=\omega(\Delta) $ with regard to the number of agents). As explained above, these mechanisms only select a $ k $-subset of agents with high in-degree, while an $ \alpha $-classification mechanism needs to select as many of worthy agents and not unworthy agents. This extra predicament shows in the results as we bound the quality of any IC mechanism away from 1 (i.e., for any $ \alpha<1 $ there is no ideal IC classification mechanism). This also shows the significance of our mechanism which, under reasonable assumptions, selects not only {\emph{good }}{agents but the }{\emph{right }}{agents.}  

The graphs in Figure~\ref{fig: main results} summarize our main findings. Graph (a) shows our bounds on the quality of any IC mechanism as a function of $ \alpha $ when $ \Delta $ is negligible with respect to the number of agents. The grey dashed line is the ideal mechanism. The upper bound (red line) shows that for any $ \alpha<1 $, there is a gap between the best possible mechanism and the ideal mechanism, and this disparity is larger for lower values of $ \alpha $. The blue line denotes the quality of our proposed mechanism, and so the green area is what we know about the possible value of the quality of the best IC mechanism. Graph (b) shows the quality of any IC mechanism as a function of $ \alpha $ when $ \Delta $ is not bounded. Here the upper and lower bounds coincide to a single line, that is, we know what is the best possible quality for any value of $ \alpha $.

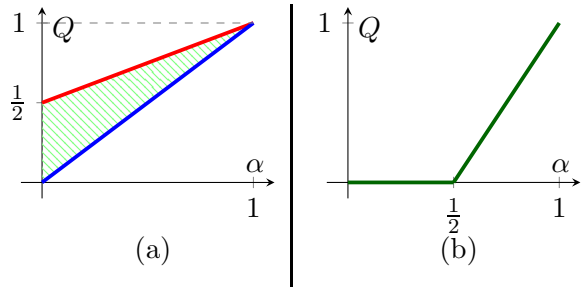
\begin{figure}[h]
\centering
	\begin{tabular}{l|l}

		\begin{tikzpicture}
		\begin{axis}[
		width=0.3\textwidth,
		height=0.25\textwidth,
		axis lines=middle,
		xlabel=$\alpha$,
		ylabel=$Q$,
		enlargelimits,
		ytick={0.5,1},
		yticklabels={$ \tfrac{1}{2} $,1},
		xtick={1},
		xticklabels={1}]
		
		\addplot[name path=F,line width=0.5mm,red,domain={0:1}] {0.5*x+0.5};
		
		\addplot[name path=G,line width=0.5mm,blue,domain={0:1}] {x};

		\addplot[pattern=north west lines, pattern color=green!50]fill between[of=F and G, soft clip={domain=0:1}];
		\addplot[name path=H,dashed,gray,domain={0:1}] {1};		
		
		\end{axis}
		\node at (50pt,-20pt) {(a)};
		\end{tikzpicture} &
		\begin{tikzpicture}
		\begin{axis}[
		width=0.3\textwidth,
		height=0.25\textwidth,		
		axis lines=middle,
		xlabel=$\alpha$,
		ylabel=$Q$,
		enlargelimits,
		ytick={1},
		yticklabels={1},
		xtick={0.5,1},
		xticklabels={$ \tfrac{1}{2} $,1}]
		\addplot[name path=F,line width=0.5mm,black!60!green,domain={0:1/2}] {0};
		
		\addplot[name path=G,line width=0.5mm,black!60!green,domain={1/2:1}] {2*x-1};
		\end{axis}
		\node at (50pt,-20pt) {(b)};		
		\end{tikzpicture} 		
	\end{tabular}
	\caption{Our bounds for the quality of an IC classification mechanism as a function of $ \alpha $. In (a) $ \Delta=o(n) $. The red/blue graphs denote our upper/lower bounds and the green area denotes the known feasible quality range. In (b) $ \Delta=n-1 $. The single green graph is the quality of any IC mechanism.}
	\label{fig: main results}
\end{figure}

\subsection{Related work}\label{sec: related work}
There have been in the past decade quite a few works on IC (sometimes called strategyproof or impartial) selection mechanisms in networks. The most similar model to ours is \cite{Kurokawa:2015:IPR:2832249.2832330}. In that paper the authors considered the $ k $-selection in a weighted $ \Delta $-regular network (that is, every agent has $ \Delta $ out-edges and $ \Delta $ in-edges), where $ k $ and $\Delta $ are constants. They considered IC mechanisms which try to optimize the average in-weights of the selected set. Their main result is a probabilistic mechanism which is optimal when $ \Delta $ is negligible with respect to $ k $.\\
Other papers on this subject only consider unweighted networks, which can be seen as a special case of the weighted networks with weights in $ \{0,1\} $. These works can be divided into two flavours:
\begin{enumerate}[label=(\alph*)]
	\item \emph{Optimisation works} which try to optimize the total in-degree of the selected agent or set of agents. Examples of this group are \citet{AFPT11,Fischer:2014:OIS:2600057.2602836,Bjelde:2017:ISP:3174276.3107922,10.1007/978-3-319-13129-0_10}. Another example is \citet{Babichenko:2018:ID:3178876.3186043} in which the authors try to optimize the progeny of the selected agent. These works differ in several parameters of the problem, such as deterministic/probabilistic mechanisms; exact selection (selected set must be of size $ k $) or inexact selection (selected set is of size at most $ k $); and the sub-family of networks considered (all networks/$ m $-regular networks/acyclic networks/etc.).
	\item \emph{Axiomatic works} which define a set of axioms and investigate the possibility/impossibility of mechanisms that fulfil maximal subsets of these axioms. Examples of this group are \citet{ECTA:ECTA1291,MACKENZIE201515,Aziz:2016:SPS:3015812.3015872}. In \cite{Altman:2008:AFR:1622655.1622669} the authors considered the possibility of complete ranking mechanisms under certain axioms.
\end{enumerate}

\emph{Paper organization.} In Section~\ref{sec: model}, we formally present our model and main results. In Section~\ref{sec: impossibilities}, we prove our impossibility (upper bound) propositions (Propositions~\ref{prp: small delta} and~\ref{prp: large delta}). These proofs rely on an extension of our model to a symmetric, probabilistic mechanism; this extension is formally defined in the beginning of Section~\ref{sec: impossibilities}. In Section~\ref{sec: mechanism}, we present and prove a mechanism with a positive quality when the number of reviews of a single agent is negligible with respect to the number of agents (Proposition~\ref{prp: deterministic mechanism}). In Section~\ref{sec: conclusions}, we conclude and discuss our results.\\

\section{Model and main results}\label{sec: model}
Let $ N=[n] $ be a set of $ n $ agents.\footnote{We assume $ n $ is large as we are generally interested in results which are asymptotic in $ n $. For the same reason we habitually drop floors and ceilings for easier reading.} We represent the interactions between the agents as a directed graph, $ G(N,E) $; thus an edge $ (x,y) $ means that agent $ x $ interacted with agent $ y $ and is allowed to review this agent. Let $ E_{in}(x), E_{out}(x) $ be the sets of in-edges and out-edges of $ x $, respectively. We assume that each agent in the network is in control of the weights of his outgoing edges. These weights, which are real numbers in the interval [-1,1],\footnote{A review of `0' is the neutral review. For example, an agent with all-zero in-edges is rated in the same way as an agent with no in-edges (see the definition of $ s(\cdot) $ in (\ref{eq: score})).\label{ftnt: 1}} represent the reviews of the source agent for the target agents. Thus, the reviews of agent $ x\in N $ for his interactions  are $ \{w_e|{e\in E_{out}(x)}\} $. After all agents submitted reviews on their interactions,\footnote{We can set a zero-weight on all interactions which were not reviewed (see Footnote~\ref{ftnt: 1}), hence we may assume w.l.o.g. that all interactions have been reviewed.} we get a weighted, directed graph; from now on we assume all the edges of $ G $ are weighted. \\
Based on the reviews that agent $ x $ received, $ \{w_e|{e\in E_{in}(x)}\} $, we define his score and his relative ranking in the system. The \emph{score} of agent $ x $ is the average of weights on $ E_{in}(x) $:
\begin{flalign}
 s(x,G)=\begin{cases}
 \dfrac{\sum_{e\in E_{in}(x)}w_{e}}{|E_{in}(x)|}, &E_{in}(x)\neq\emptyset,\\
 0,&E_{in}(x)=\emptyset.
 \end{cases}\label{eq: score}
\end{flalign}
 The \emph{ranking} of agent $ x $ is the number of agents who strictly have a better score than him: $ r(x,G)=|\{y\in N| s(y,G)>s(x,G) \}| $.\footnote{From now on, when $ G $ is clear from the context, we just write $ s(x), r(x) $.} For a given real parameter $ \alpha\in(0,1) $ we consider as \emph{worthy} the subset of agents who are in the top $ \alpha $-ranking:\footnote{We think of $ \alpha $ as the selectiveness of the system:  a lower value for $ \alpha $ means that the tag of being worthy is more prestigious.} 
\[ I_\alpha(G)=\{ x\in N| r(x)<\alpha n\ \}. \]
Notice that the size of $ I_\alpha(G) $ is at least $ \alpha n $, but might be higher in case of ties. For instance, in the empty graph, $ I_\alpha(G(N,\emptyset))=N $ for all $ \alpha $. We denote by $ \overline{I_{\alpha}}(G)=N\backslash I_\alpha(G) $, the subset of \emph{unworthy} agents. \\
Our goal is to offer an IC mechanism which selects a set which is as similar as possible to the subset of worthy agents. Formally, let $ \cl{G}(n) $ be the family of all [-1,1]-weighted, directed networks on $ n $ nodes and let $ P(N) $ be the power set of $ N $. 
\begin{defn}\label{dfn: deterministic mechanism}
   	A \emph{classification mechanism} is a function $ M:\cl{G}(n)\to P(N) $.
\end{defn}
The set $ M(G) $ is the subset of agents which the mechanism $ M $ classifies as worthy in the network $ G $. We denote by $ \overline{M}(G)=N\backslash M(G) $ the subset classified as unworthy by the mechanism. Notice that the definition of a mechanism depends on $ n $ through its dependence on $ \cl{G}(n) $ and $ N $. We abuse notation and regard a single mechanism $ M $ as if it represents a series of mechanisms---one for every natural $ n $. \\
The IC requirement means that an agent's classification is not influenced by his own reviews; that is, changing the weights on the out-edges of an agent does not alter his classification.
\begin{defn}
	A classification mechanism $ M $ is \emph{incentive-compatible} if for every $ n\in\NN, G, G'\in \cl{G}(n), x\in N $, such that: $ E(G)=E(G') $ and $ \forall e\in E\backslash E_{out}(x) $, $ w_e(G)=w_e(G') $, \[ x\in M(G)\iff x\in M(G').\]
\end{defn}
To define a measure for the quality of the mechanism, we first define a measure of coincidence between $ M(G) $ and $ I_\alpha(G) $:
\begin{flalign}
&\cl{C}(M(G), I_\alpha(G))=\dfrac{1}{n}\sum_{x\in N}\begin{cases}
1,&x\in (M(G)\cap I_\alpha(G))\cup (\overline{M}(G)\cap\overline{I_\alpha}(G)),\\
-1,&\text{otherwise}.\end{cases}\label{eq: coincidence measure}
\end{flalign}
In other words, $ \cl{C}(M(G), I_\alpha(G) $ gives one point for every agent that $ M(G) $ classified correctly and takes one point for every agent which was classified erroneously; the result is normalised by the number of agents.\footnote{Other measures might be appropriate for different applications. We discuss two of these alternatives in Section~\ref{sec: conclusions}. The main conclusions from our results stay the same in these variations.} This measure can be somewhat simplified to get,\footnote{The operator $ \oplus $ is the ``exclusive or''. }

\begin{flalign*}
\cl{C}(M(G), I_\alpha(G))&=\dfrac{1}{n}(|((M(G)\cap I_\alpha(G))\cup(\overline{M}(G)\cap \overline{I_\alpha}(G)))|-|M(G)\oplus I_\alpha(G)|  )&\\
&=\dfrac{1}{n}(|N\backslash (M(G)\oplus I_\alpha(G))|-|M(G)\oplus I_\alpha(G)| )
=1-\dfrac{2}{n}|M(G)\oplus I_\alpha(G)|.&
\end{flalign*}

Our main theorems imply that the possibility of an IC mechanism which guarantees a fixed level of coincidence depends on two parameters. The first is $ \alpha $. The second is the maximal out-degree in the network (i.e., the maximal reviews an agent can issue), denoted by $ \Delta $. Intuitively, if $ \Delta $ is high, then an unworthy agent might use his influence on the score of $ \Delta $ worthy agents to improve his ranking and be considered worthy. If $ \alpha $ is relatively low, then these manipulations might influence a large portion (or all) of the worthy-classified agents, which makes it harder to find an IC mechanism with high coincidence measure. Our results will show that this intuition is indeed correct.
Let $ \cl{G}(n,\Delta) $ be the family of all networks on $ n $ nodes with maximal out-degree $ \Delta $. The quality of a mechanism for given $ \alpha, \Delta $, is the limit when $ n $ goes to infinity of the worst case of the coincidence measure:
\begin{flalign}
Q_{\alpha,\Delta}(M)=\lim\limits_{n\to\infty}\;\min_{G\in\cl{G}(n,\Delta)}\cl{C}(M(G),I_\alpha(G)).\label{eq: quality}
\end{flalign}
We will prove upper and lower bounds on $ Q_{\alpha,\Delta} $ for any IC classification mechanism.

\subsection{Main results}
We start by defining two trivial IC mechanisms. `Trivial' means that they classify the nodes without any regard to the edges' weights.\\
 Let $ M_N $ be the complete mechanism, i.e., $ M_N(G)=N $ for all $ G\in\cl{G}(n) $. 
\begin{prop}\label{prp: complete mechanism}
	For any $ \alpha,\Delta $, $ Q_{\alpha,\Delta}(M_N)=2\alpha-1 $.
\end{prop}
\begin{proof}
	Since $ |N\oplus I_\alpha(G)|=|\overline{I_\alpha}(G)|=n-|I_\alpha(G)|\leq n(1-\alpha) $, we get that for any graph,
	\begin{flalign*}
	\cl{C}(M_N(G),I_\alpha(G))=1-\dfrac{2}{n}|N\oplus I_\alpha(G)|&\geq 1-\dfrac{2n(1-\alpha)}{n}=2\alpha-1.
	\end{flalign*}
\end{proof}
Our second trivial mechanism, $ M_{1/2} $, selects every node to be worthy with probability 1/2, independently. We use here the concept of a probabilistic mechanism intuitively; in the beginning of Section~\ref{sec: impossibilities}, we formally extend our model to include this kind of mechanism. Mechanism $ M_{1/2} $ correctly classifies nodes $ x $ with probability 1/2, hence $ Q_{\alpha,\Delta}(M_{1/2})=0 $ for all $ \alpha, \Delta $.\\
Our first result is a strong impossibility, saying that when $ \Delta $ is large, one of the two trivial mechanisms is the best possible.
\begin{thm}\label{thm: large delta}\quad
	\begin{enumerate}[label=(\alph*),leftmargin=10mm]
		\item For $ \alpha\geq \dfrac{1}{2} $ and $ \Delta\geq 2(1-\alpha)n $, $ Q_{\alpha,\Delta}=2\alpha-1 $.
		\item For $ \alpha<\dfrac{1}{2} $ and $ \Delta=(1-o(1))n $, $ Q_{\alpha,\Delta}=0 $.
	\end{enumerate}
\end{thm}
\begin{proof}
	This is a direct consequence of Proposition~\ref{prp: complete mechanism}, our observation for $ M_{1/2} $ and Proposition~\ref{prp: large delta}.
\end{proof}
Our second result is a non-trivial, yet quite natural, mechanism with a better quality than the complete mechanism, provided $ \Delta=o(n) $. The idea behind this mechanism is to recognize three subsets of agents: absolutely worthy, absolutely unworthy, and borderline. The first two subsets contain those agents who will not be able to change their classification, no matter what their reviews will be. The fact that $ \Delta $ is negligible guarantees that the absolutely worthy and absolutely unworthy sets will include a large portion of the true worthy and unworthy subsets, respectively. The mechanism then classifies the absolutely worthy agents as worthy and the absolutely unworthy as unworthy. If we allow the mechanism to be probabilistic, we can select each of the borderline agents to be worthy with probability 1/2; this strategy assures that these agents will not hurt the quality (but will not help it either). We also provide a deterministic version of this mechanism which always classifies correctly almost half of the borderline agents. The following theorem summarises our knowledge when $ \Delta=o(n) $.
\begin{thm}\label{thm: small delta}
	For any $ \alpha $ and $ \Delta=o(n) $, $ \alpha\leq Q_{\alpha,\Delta}=\dfrac{1}{2}(1+\alpha) $.
\end{thm}
\begin{proof}
	The lower bound comes from an analysis of the above-described probabilistic mechanism; see Proposition~\ref{prp: probabilistic mechanism}. We also prove the existence of a deterministic version of this mechanism in Proposition~\ref{prp: deterministic mechanism}. The upper bound is proved in Proposition~\ref{prp: small delta}.
\end{proof}

\section{Impossibilities}\label{sec: impossibilities}
As promised, we start by extending our model to include probabilistic mechanisms. A probabilistic mechanism assigns each node a probability of being worthy. 
\begin{defn}\label{dfn: probabilistic mechanism}
	A \emph{probabilistic classification mechanism} is a function $ M_p:N\times\cl{G}(n)\to[0,1] $.
\end{defn}
To get a concrete selected set from a probabilistic mechanism, we select each node independently with his assigned probability. In other words, the probability of subset $ X\subseteq N $ to be selected under mechanism $ M_p $ in the graph $ G $ is 
\begin{flalign}
\Pr[X|M_p(G)]&=\prod_{x\in X}M_p(x,G)\prod_{x\notin X}(1-M_p(x,G)).\label{eq: probabilistic coincidence}
\end{flalign}
The IC requirement translates to the requirement that an agent cannot influence his own selection probability.
\begin{defn}\label{dfn: probabilistic IC}
	A probabilistic classification mechanism $ M_p $ is \emph{incentive-compatible} if for every $ n\in\NN, G, G'\in \cl{G}(n), x\in N $, such that: $ E(G)=E(G') $ and $ \forall e\in E\backslash E_{out}(x) $, $ w_e(G)=w_e(G') $, \qquad$  M_p(x,G)= M_p(x,G') $.
\end{defn}
The coincidence of $ M_p(G) $ with $ I_\alpha(G) $ is naturally extended using expectation over the selected set:
\begin{flalign*}
\cl{C}(M_p(G),I_\alpha(G))&=\sum_{X\in P(N)}\Pr[X|M_p(G)]\cl{C}(X,I_\alpha(G))\\
&=\dfrac{1}{n}\sum_{X\in P(N)}\Pr[X|M_p(G)]\cdot
\sum_{x\in N}\begin{cases}
1,&x\in (X\cap I_\alpha(G))\cup (\overline{X}\cap\overline{I_\alpha}(G))\\
-1,&\text{otherwise}.\end{cases}
\end{flalign*}
Changing the summation order and inserting~(\ref{eq: probabilistic coincidence}) we get,
\begin{flalign*}
&\cl{C}(M_p(G),I_\alpha(G))=\\
&\dfrac{1}{n}\sum_{x\in N}\sum_{X\in P(N\backslash\{x\})}
\prod_{y\in N\backslash\{x\}} \begin{cases}
M_p(y,G),&y\in X\\
1-M_p(y,G),&y\notin X
\end{cases}\cdot\begin{cases}
M_p(x,G)-(1-M_p(x,G)),&x\in I_\alpha(G)\\
(1-M_p(x,G))-M_p(x,G),&x\in\overline{I_\alpha}(G).
\end{cases}
\end{flalign*}
Since
$ \sum\limits_{X\in P(N\backslash\{x\})}\prod\limits_{y\in N\backslash\{x\}} \begin{cases}
M_p(y,G),&y\in X\\
1-M_p(y,G),&y\notin X
\end{cases}=1 $, 
\begin{flalign}
&\cl{C}(M_p(G),I_\alpha(G))=\dfrac{1}{n}\sum_{x\in N}(2M_p(x,G)-1)\cdot\begin{cases}
1,&x\in I_\alpha(G)\\
-1,&x\in\overline{I_\alpha}(G)
\end{cases}\label{eq: 0}\\
&=\dfrac{\overline{I_\alpha}(G)-I_\alpha(G)}{n}+\dfrac{2}{n}\sum_{x\in N}M_p(x,G)\cdot\begin{cases}
1,&x\in I_\alpha(G)\\
-1,&x\in\overline{I_\alpha}(G).
\end{cases}\label{eq: 1}
\end{flalign}

The quality of a probabilistic mechanism can now be defined exactly as in (\ref{eq: quality}):
\begin{flalign*}
Q_{\alpha,\Delta}(M_p)=\lim\limits_{n\to\infty}\;\min_{G\in\cl{G}(n,\Delta)}\cl{C}(M_p(G),I_\alpha(G)).
\end{flalign*}
From Definitions~\ref{dfn: probabilistic mechanism} and~\ref{dfn: probabilistic IC} and the coincidence definition above, it is clear that the IC (deterministic) mechanisms for which we initially defined our problem are a special case of IC probabilistic mechanisms. Hence we may prove our upper bounds (i.e., impossibilities) for probabilistic mechanisms. We furthermore show that it is enough to consider a subfamily of \emph{symmetric}, IC, probabilistic mechanisms.
\begin{defn}\label{dfn: symmetry}
	A probabilistic mechanism $ M_p $ is \emph{symmetric} if for any network $ G\in\cl{G}(n) $ and two isomorphic nodes $ x,y\in N $, $ M_p(x,G)=M_p(y,G) $.
\end{defn}
\begin{clm}
	Let $ M_p $ be any IC, probabilistic, classification mechanism. Then there is an IC, symmetric, probabilistic classification mechanism $ M_p' $ with $ Q_{\alpha,\Delta}(M_p')\geq Q_{\alpha,\Delta}(M_p) $.
\end{clm}
\begin{proof}
	Let $ S(N) $ be the set of permutations over $ N $. For $ \pi\in S(N) $, let $ G_\pi $ be the graph which is isomorphic to $ G $ under the automorphism defined by $ \pi $. We define the mechanism $ M_p' $:
	\begin{flalign}
	M_p'(x,G)=\dfrac{1}{n!}\sum_{\pi\in S(N)}M_p(\pi(x),G_\pi).\label{eq: 2}
	\end{flalign}
	Mechanism $ M_p' $ is clearly IC, since $ M_p $ is IC and otherwise the calculation above is irrelevant of any of the weights in $ G $. It is also symmetric, since for two isomorphic nodes, $ x,y $, the following two sets of the couples are exactly the same: \[ \{(\pi(x), G_\pi)|\pi\in S(N) \}=\{(\pi(y), G_\pi)|\pi\in S(N) \}. \]
	
	It remains to show that the quality of $ M_p' $ is at least that of $ M_p $. Planting (\ref{eq: 1}) into (\ref{eq: 2})  and changing the summation order we get that for any $ G $,
	\begin{flalign*}
	&\cl{C}(M_p'(G),I_\alpha(G))=\\
	&\dfrac{\overline{I_\alpha}(G)-I_\alpha(G)}{n}+\dfrac{2}{n\cdot n!}\sum_{x\in N}\sum_{\pi\in S(N)}M_p(\pi(x),G_\pi)\cdot\begin{cases}
	1,&x\in I_\alpha(G_\pi)\\
	-1,&x\in\overline{I_\alpha}(G_\pi)
	\end{cases}\\
	=&\dfrac{\overline{I_\alpha}(G)-I_\alpha(G)}{n}
	+\dfrac{1}{n!}\sum_{\pi\in S(N)}\left[\dfrac{2}{n}\sum_{x\in N}M_p(\pi(x),G_\pi)\cdot\begin{cases}
	1,&x\in I_\alpha(G_\pi)\\
	-1,&x\in\overline{I_\alpha}(G_\pi)
	\end{cases}.\right]
	\end{flalign*}
	Let $ Q=Q_{\alpha,\Delta}(M_p) $. For any $ \epsilon>0 $, there is $ n_0 $ such that for all $ n>n_0 $ and for all $ G\in\cl{G}(n,\Delta) $, $ \cl{C}(M_p(G, I_\alpha(G)))\geq Q-\epsilon $, which means that 
	\begin{flalign*}
	\dfrac{2}{n}\sum_{x\in N}M_p(x,G)&\cdot\begin{cases}
	1,&x\in I_\alpha(G)\\
	-1,&x\in\overline{I_\alpha}(G)
	\end{cases}\geq Q-\epsilon-\dfrac{\overline{I_\alpha}(G)-I_\alpha(G)}{n}.
	\end{flalign*}
	Hence, $ \cl{C}(M_p'(G),I_\alpha(G))\geq Q-\epsilon $. Since this is true for any $ \epsilon $, we get that $ Q_{\alpha,\Delta}(M_p')\geq Q $.
\end{proof}

We are now ready to prove our two impossibility propositions which imply the upper bounds of Theorems~\ref{thm: large delta} and~\ref{thm: small delta}.

\begin{prop}\label{prp: small delta}
	For all $ \alpha, \Delta $, $ Q_{\alpha,\Delta}\leq\dfrac{1}{2}(1+\alpha) $.
\end{prop}
\begin{proof}
	  Let $ v $ be a distinct node. Partition $ N\backslash\{v\} $ into three sets, $ A, B, C $, of sizes $ \dfrac{(1-\alpha)n}{2}, \dfrac{(1-\alpha)n}{2}, \alpha n-1 $, respectively. Let $ G $ be the network in which every node in $ A\cup B $ has an out-edge to $ v $, and $ C $ is a cycle. Set the weights on the edges from $ A $ to $ v $ to be 1, the weights on the edges from $ B $ to $ v $ to be $ -1+\dfrac{1}{(1-\alpha)n} $, and the weights on $ C $ to be 1. For any $ a\in A, b\in B, c\in C $ their scores are: $ s(c)=1 $, $ s(a)=s(b)=0 $, and $ s(v)=\dfrac{|A|+|B|(-1+1/(1-\alpha)n)}{|A|+|B|}=\dfrac{1}{2(1-\alpha)n}>0 $. Hence $ I_\alpha(G)=C\cup\{v\} $. Let $ M $ be an IC, probabilistic and symmetric classification mechanism. By symmetry, we may denote $ M(a,G)=\mu $ for any $ a\in A $. Using (\ref{eq: 0}) we get that,
	  \begin{flalign}
	  \cl{C}(M(G),I_\alpha(G))\leq \dfrac{1}{n}((1-2\mu)|A|+|B|+|C|+1)
	  =1-\mu(1-\alpha).\label{eq: first bound}
	  \end{flalign}
	  Now choose a distinct node $ a_0\in A $ and change the weight on its out-edge to $ v $ to be $ -1+\dfrac{1}{(1-\alpha)n} $; we refer to this network as $ G' $. The score of $ v $ has dropped in $ \dfrac{2-1/(1-\alpha)n}{|A|+|B|}=\dfrac{2-1/(1-\alpha)n}{(1-\alpha)n}>\dfrac{1}{2(1-\alpha)n} $, for $ n $ large enough. Hence $ s(v,G')<0  $. Since the rest of the scores are the same in $ G $ and $ G' $, we get that $ I_\alpha(G')=N\backslash\{v\} $. By IC, $ M(a_0,G')=M(a_0,G)=\mu $, and by symmetry $ M(b,G')=\mu $ for any $ b\in B $. We now get	,
	  \begin{flalign}
	  \cl{C}(M(G'),I_\alpha(G'))\leq \dfrac{1}{n}(|A|+(2\mu-1)|B|+C+1)
	  =1-(1-\mu)(1-\alpha)=\alpha+\mu(1-\alpha).\label{eq: second bound}
	  \end{flalign}
	  From (\ref{eq: first bound})  and (\ref{eq: second bound}), $ Q_{\alpha, \Delta}\leq\min\{1-\mu(1-\alpha), \alpha+\mu(1-\alpha) \} $. Comparing the two terms to find the optimal value for $ \mu $ we find that
	  \begin{flalign*}
	  1-\mu(1-\alpha)=\alpha+\mu(1-\alpha)\iff \mu=\dfrac{1}{2}\\
	  \Longrightarrow Q_{\alpha,\Delta}\leq 1-\dfrac{1-\alpha}{2}=\dfrac{1+\alpha}{2}.\qedhere
	  \end{flalign*}
\end{proof}

\begin{prop}\label{prp: large delta}
	Let $ m=\min\left\{2(1-\alpha),\dfrac{\Delta}{n}\right\}  $. Then 
	$ Q_{\alpha,\Delta}\leq 1-m. $
\end{prop}
\begin{proof}
	Consider two cases:\\
	
		\emph{\underline{Case I:} $ m\geq 2\alpha $.}\\	
		Let $ A,B\subseteq N $ be two disjoint subsets of size $ nm/2$. Let $ G $ be the graph in which $ A\cup B $ is a clique and there are no other edges.
	We set the weights on all the out-edges of nodes in $ A $ to be 0, and the weights on all the out-edges of nodes in $ B $ to be 1. Hence every $ a\in A $ has a score of $ s(a)=\dfrac{|B|}{|A\cup B|} $, every $ b\in B $ has a score of $ s(b)=\dfrac{|B|-1}{|A\cup B|} $ and the score of every $ c\notin A\cup B $ is 0. Since $ |A|=\dfrac{mn}{2}\geq\alpha n$, $ I_\alpha(G)=A $. Let $ M $ be an IC, probabilistic, symmetric mechanism. By the symmetry of $ M $, all the vertices of $ B $ get the same probability. Denote it $ \mu $. Let $ b $ be a distinct vertex in $ B $. Let $ G' $ be the graph we get when we nullify all the weights on the outgoing edges of $ b $. Since $ b $ is now isomorphic to the vertices in $ A $, we get by IC and the symmetry of $ M $ that $ \forall a\in A $, $ M(a,G')=M(b,G')=M(b,G)=\mu $. We see that there is a trade-off in the value of $ \mu $; on the one hand, we need it to be high if we want a good quality on $ G' $, but on the other hand, it needs to be low for a good quality on $ G $. The idea of the proof is that for $ n $ large enough, we can repeat this process and show that in essence all the vertices in $ A\cup B $ in the graph $ G $ (or at least in one of the graphs we get from $ G $ after a negligible number of steps) should have approximately the same probability. This implies that 
	\begin{flalign*}
	Q_{\alpha,\Delta}\leq\dfrac{1}{n}( |A|(2\mu-1)+|B|(1-2\mu)+|N\backslash(A\cup B)|)
	=\dfrac{2\mu-1}{n}(|A|-|B|)+1-\dfrac{|A\cup B|}{n}=1-m.	
	\end{flalign*}
	To make this claim precise, suppose for contradiction that $ Q_{\alpha,\Delta}(M)\geq1-m+\epsilon $ for some $ \epsilon>0 $. For $ k\leq X $, $ X $ to be found, we let $ A^k, B^k\subseteq N $ be two disjoint subsets with sizes $ |A^k|=mn/2+k $, $ |B^k|=mn/2-k $. Define the graph $ G^k $ in which $ A^k\cup B^k $ is a clique, and the weights on the outgoing edges of vertices in $ A^k $ are all 0, and the weights on the outgoing edges of vertices in $ B^k $ are all 1. Notice the following:
		\begin{enumerate}[label=\alph*)]
			\item The vertices in $ A^k  $ are symmetric, and so are the vertices in $ B^k $.
			\item $ I_\alpha(G^k)=A^k $.
			\item If we nullify the outgoing edges of one of the nodes $ b\in B^{k} $ then we get the graph $ G^{k+1} $in which $ b\in A^{k+1} $.
		\end{enumerate}
		By the symmetry of $ M $, we may denote for any $ k $, $ \mu_a^k=M(a,G^k) $ for all $ a\in A^k$, and $ \mu_b^k=M(b,G^k) $ for all $ b\in B^k $. By (c) and IC of $ M $, $ \mu_b^k=\mu_a^{k+1} $. By the assumption on the quality of $ M $:
		\begin{flalign*}
		&\dfrac{1}{n}(|A^k|(2\mu_a^k-1)+|B^k|(1-2\mu_b^k)+|N\backslash (A^k\cup B^k)|)\geq 1-m+\epsilon,\\
		&\dfrac{1}{n}((mn/2+k)(2\mu_a^k-1)+(mn/2-k)(1-2\mu_b^k)+(1-m)n)\geq1-m+ \epsilon,\\
		&(mn+2k)\mu_a^k-(mn-2k)\mu_b^k-2k\geq \epsilon n.
		\end{flalign*}
		Summing up this inequality for $ 0\leq k\leq X $ and substituting $ \mu_b^k $ for $ \mu_a^{k+1} $ we get,
		\begin{flalign*}
		mn\mu_a^0-(mn-2X)\mu_b^{X}+2\sum_{k=1}^{X}(2k-1)\mu_a^k-{X(X+1)}
		\geq X\epsilon n.
		\end{flalign*}
		Let $ \mu_{max}=\max\limits_{1\leq k\leq X}\mu_a^k $. If $ X=o(n) $, then for $ n $ large enough $ mn\geq 2X $. We strengthen the inequality by removing the negative terms on the left-hand side, and replacing all the $ \mu_a^k $ with $ \mu_{max} $:
		\begin{flalign*}
			\mu_{max}\left (mn+2\sum_{k=1}^{X}(2k-1)\right )\geq X\epsilon n\\
			\mu_{max}(mn+2X^2)\geq X\epsilon n\\
			\Longrightarrow \mu_{max}\geq \dfrac{\epsilon X}{m+2X^2/n}.
		\end{flalign*}
		Now taking $ X=2m/\epsilon $ we get that for $ n $ large enough $ \mu_{\max}>1 $, which is a contradiction.\\
		
		\emph{\underline{Case II:} $ m< 2\alpha $}\\
		The proof is very similar. We define three subsets, $ A, B, C $ of size $ nm/2, nm/2 $ and $ (\alpha-m/2)n $, respectively.\footnote{Since we assume $ m<2\alpha $, $ |C|>0 $.} Let $ G $ be the graph in which $ A\cup B $ is a clique and $ C $ is a cycle. Again we set all the weights on the out-edges of the nodes in $ A $ to be 0, and weights on all the out-edges of nodes in $ B $ to be 1. The edges in the cycle in $ C $ also get a weight of 1. Thus now $ s(c)=1>s(a)>s(b) $ for any $ c\in C, a\in A, b\in B $. Since $ |A\cup C|=\alpha n $, we have $ I_\alpha(G)=A\cup C $. Let $ M $ be an IC, probabilistic, symmetric mechanism. Using the same technique as before, we get that all the vertices in $ A\cup B $ should get the same probability, which we call $ \mu $. Since $ I_\alpha(G)=A\cup C $ we can bound as before:
		\begin{flalign*}
		&Q_{\alpha,\Delta}(M)\leq\dfrac{1}{n}(|A|(2\mu-1)+|B|(1-2\mu)+|N\backslash(A\cup B)|)
		=1-m.
		\end{flalign*}  
		The formal argument is precisely the same.
\end{proof}

\section{Non-trivial mechanism}\label{sec: mechanism}
In this section we show the existence of a mechanism with quality $ \alpha $, provided $ \Delta=o(n) $. Specifically, we will show the following:
\begin{prop}\label{prp: deterministic mechanism}
	For any $ \alpha $ and any $ \Delta\leq\min \{\tfrac{1}{3}\alpha n, \tfrac{1}{3}(1-\alpha)n \} $, there is an IC classification mechanism with quality $ \alpha-\dfrac{3\Delta}{n} $.
\end{prop}
For better exposition, we start by presenting a probabilistic mechanism which is slightly better.
\begin{prop}\label{prp: probabilistic mechanism}
	For any $ \alpha $ and any $ \Delta\leq\alpha n $, there is an IC probabilistic classification mechanism with quality $ \alpha-\dfrac{\Delta}{n} $.
\end{prop}
For a graph $ G $ and $ x\in G $, denote by $ G_x $ the graph we get when we set all the weights on the out-edges of $ x $ to be -1. Let $ \beta=\alpha-\dfrac{\Delta}{n} $. The idea is to partition $ N $ into three subsets:
\begin{flalign*}
W(G)&=\{x\in N| x\in I_\beta(G_x) \},\\
U(G)&=\{x\in N| x\in\overline{I_\alpha}(G_x) \},\\
B(G)&=N\backslash (W\cup U) =\{x\in N| x\in I_\alpha(G_x)\backslash I_\beta(G_x)  \}.
\end{flalign*}
Notice that the definitions of $ W, U, B $ are IC, in the sense that the sorting of $ x $ to one of these sets does not depend on his own reviews (since we fixed all his reviews to -1 before choosing his set). If $ x\in W $, then his ranking in $ G_x $ is less than $ \beta n $: $ r(x, G_x)<\beta n $. The real reviews of $ x $ may increase the score of at most $ \Delta $ agents, which means that $ r(x,G)\leq r(x,G_x)+\Delta<\beta n+\Delta=\alpha n $; hence $ x\in I_\alpha(G) $. We have proved that $ W\subseteq I_\alpha(G) $. We think of $ W $ as the \emph{absolutely worthy} agents. Similarly, the set $ U $ is the set of \emph{absolutely unworthy} agents: for $ x\in U $, $ r(x,G)\geq r(x,G_x)\geq\alpha n $ (increasing his reviews can only hurt his relative ranking), and $ U\subseteq \overline{I_\alpha}(G) $. The set $ B $ contains all the borderline agents: these are the agents which we are not sure how to classify. We define the following probabilistic mechanism:
\begin{flalign*}
M_p(x,G)=\begin{cases}
1,&x\in W(G),\\
0,&x\in U(G),\\
1/2,&x\in B(G).
\end{cases}
\end{flalign*}
Since $ W, U, B $ are IC, this mechanism is IC. The mechanism is always correct in the classification of the agents in $ W\cup U $. We set the probability of agents in $ B $ to be 1/2 so that the expected contribution of the agents in $ B $ to the coincidence measure is zero. We get that for any $ G\in\cl{G}(n,\Delta) $ with $ \Delta\leq\alpha n $,
\begin{flalign*}
\cl{C}(M_p(G),I_\alpha(G))=\dfrac{1}{n}(|W(G)|+|U(G)|)\geq \dfrac{|W(G)|}{n}
\geq \beta
=\alpha-\dfrac{\Delta}{n}.
\end{flalign*}
This completes the proof of Proposition~\ref{prp: probabilistic mechanism}. \\
The idea for the deterministic mechanism is similar. This mechanism selects as worthy all the agents in $ W $ and as unworthy all the agents in $ U $. For the agents in $ B $ we need to find an IC, deterministic way to classify them such that in every network about half of them are rightly classified. Notice first that if $ |U|\geq |B|-2{\Delta} $ then 
\begin{flalign*}
\cl{C}(M(G),I_\alpha(G))\geq\dfrac{1}{n}\left(|W|+|U|-|B| \right) \geq \dfrac{1}{n}(\beta n-2\Delta)
=\alpha-\dfrac{3\Delta}{n}.
\end{flalign*}
Let $ \cl{L}(B) $ be a linear ordering of $ B $ according to $ r(\cdot,G) $ and breaking ties lexicographically.\footnote{That is, for $ x,y\in B $, $ x\succ_{\cl{L}}y\iff (r(x)>r(y))\vee ((r(x)=r(y))\wedge (x<y)) $.} 
Let $ B^+ $ be the top half nodes in $ B $ according to $ \cl{L} $.
For $ x\in B $ we denote by $ B_x $ the set $ B $ in $ G_x $.\footnote{To be clear, for $ x,y\in N $, let $ G_{x,y} $ be the graph in which we set all the weights on the outgoing edges of $ x $ and $ y $ to -1. Then $ B_x(G)=B(G_x)=\{y\in N| y\in I_\alpha(G_{x,y})\backslash I_{\beta}(G_{x,y}) \} $.} Similarly, let $ B_x^+=B^+(G_x) $. We now complete our definition of mechanism $ M $ by setting for every $ x\in B $, $ x\in M(G)\iff x\in B_x^+ $. That is, $ x\in B $ is accepted if and only if it is ranked in the top half of $ B_x $ when breaking ties lexicographically. The mechanism does not use the weights on the outgoing edges of $ x $ to determine its classification, hence it is IC. 
We complete the proof of Proposition~\ref{prp: deterministic mechanism} with the following lemma.
\begin{lem}\label{lem: deterministic mechanism}
	Suppose that $ \Delta\leq \dfrac{1}{3}(1-\alpha) n $.	For any graph $ G\in\cl{G}(n, \Delta) $ with $ |U|< |B|-2{\Delta} $ at least $ \dfrac{1}{2}(|B\cup U|-\Delta) $ of the agents in $ B\cup U $ are classified correctly.
\end{lem}
\noindent Using this lemma we get that for any such graph,
\begin{flalign*}
	\cl{C}(M(G),I_\alpha(G))&\geq \dfrac{1}{n}\left ( |W|+\dfrac{1}{2}(|B\cup U|-\Delta)-\dfrac{1}{2}(|B\cup U|+\Delta)\right ) \\
	&\geq \dfrac{1}{n}\left ( |W|-\Delta\right )\geq\dfrac{1}{n}\left(\beta n -\Delta\right)=\alpha-\dfrac{2\Delta}{n},
\end{flalign*}
as required.
\begin{proof}[Proof of Lemma~\ref{lem: deterministic mechanism}]
	Consider two cases.\\
	\emph{\underline{Case I:} $|I_\alpha(G)|\geq \alpha n+2\Delta $}.\\
	 In this case $ B\subseteq I_\alpha\backslash I_\beta $, since if $ x\notin I_\alpha(G) $, then in $ G_x $ there are at least $ |I_\alpha|-\Delta\geq \alpha n+\Delta $ agents with a higher score than $ s(x) $; hence $ x\notin I_\alpha(G_x) $. Moreover, for any $ x\in B $, $ B_x\subseteq I_\alpha(G) $, since for any $ y\notin I_\alpha $, there are in $ G_{x,y} $ at least $ |I_\alpha|-2\Delta\geq\alpha n $ agents with a higher score than $ s(y) $. Let $ B^* $ be the top $ \tfrac{1}{2}(|B|-\Delta) $ ranked agents in $ B $ according to $ \cl{L}(B) $.
	  We claim that the agents in $ B^* $ are all classified by our mechanism to be worthy, which is the correct classification (since $ B\subset I_\alpha(G) $, as explained). The reason is that when we set the weight on an edge $ (x,y) $ with $ x\in B^* $ to -1, we add at most one agent to $ B $ which is ranked above $ x $ (this might happen when $ y\in W $), or we remove at most one agent from $ B $ which is ranked below $ x $ (this might happen when $ y\in B $ and is ranked below him). Thus $ x $ must be in the top-half ranked agents in $ B_x $. Since all the agents in $ U $ are correctly classified, we get that at least $ |B^*|+|U|=\tfrac{1}{2}(|B|+2|U|-\Delta) $ are classified correctly.\\
	\emph{\underline{Case II:} $ |I_\alpha(G)|< \alpha n+2\Delta $}.\\	 
	Let $ B_* $ be the bottom $ \tfrac{1}{2}(|B|-|U|)-\Delta $ ranked agents in $ B $ according to $ \cl{L}(B)$.\footnote{Notice that $ |B_*| > 0 $ due to our assumption on $ |U| $. } Since $ |B_*\cup U|=\tfrac{1}{2}(|B|+|U|)-\Delta\leq\tfrac{1}{2}(1-\beta)n-\Delta=\tfrac{1}{2}(1-\alpha)n-\tfrac{1}{2}\Delta $, and $ |\overline{I_\alpha}(G)|=n-I_\alpha(G)>(1-\alpha)n-2\Delta\geq \tfrac{1}{2}(1-\alpha)n-\tfrac{1}{2}\Delta $, we get that $ B_*\subseteq \overline{I_\alpha}(G) $. It is therefore enough to show that for any $ x\in B_* $, $ x\notin B_x^+ $. Indeed, lowering the weights on the out-edges of $ x $ can advance $ x $ in the ranking of $ B $ by at most $ \Delta $ places (if the edges are to nodes in $ B $ that are ranked higher than $ x $), and it might also add all the nodes of $ U $ to $ |B| $; in any case $ x $ will still be in the bottom half of $ B_x $.
\end{proof}

\section{Conclusions}\label{sec: conclusions}
We have introduced a generic model for the classification of agents to worthy and unworthy according to their own reviews of each other. We draw two general conclusions regarding the existence of an IC classification mechanism:
\begin{enumerate}
	\item If $ \Delta $ is large, there is no good mechanism in the sense that the best mechanisms do not take into consideration the reviews.
	\item If $ \Delta $ is negligible with respect to the number of agents, there is a mechanism with a positive quality, but not with quality 1.
\end{enumerate}
Our measure for the coincidence between the selected set and the true worthy agents (\ref{eq: coincidence measure}) assumes that every classification/misclassification has the same value/price. In other applications it makes more sense to consider only the classification/misclassification of the worthy agents, which leads to the following measure:
\begin{flalign*}
\cl{C}(M(G),I_\alpha(G))=\dfrac{1}{|I_\alpha(G)|}\sum_{x\in M(G)}\begin{cases}
1,&x\in I_\alpha(G),\\
-1,&x\notin I_\alpha(G).
\end{cases}
\end{flalign*}
Yet another measure might consider all the agents but normalise the classification/misclassification of an agent according to his true set:\footnote{Though here we need to assume that $ \overline{I_\alpha}(G)  $ is not empty, or define the measure for this case separately.}
\begin{flalign*}
&\cl{C}(M(G),I_\alpha(G))=\dfrac{|M(G)\cap I_\alpha(G)|}{2|I_\alpha(G)|}+\dfrac{|\overline{M}(G)\cap \overline{I_\alpha}(G)|}{2|\overline{I_\alpha}(G)|}.
\end{flalign*}
We mention here that using these measures does not change our two conclusions above; in other words, our conclusions are intrinsic in the problem and not the result of a specific measure. The exact claims and proofs can be found in Appendix~\ref{sec: apx}.

%% file: classification-appx.tex
\appendix

\section{Two alternative measures}\label{sec: apx}

Consider the following two alternative coincidence measures and the appropriate quality measures. 
\begin{flalign*}
\cl{C}'(M(G),I_\alpha(G))&=\dfrac{1}{|I_\alpha(G)|}\sum_{x\in M(G)}\begin{cases}
1,&x\in I_\alpha(G)\\
-1,&x\notin I_\alpha(G)
\end{cases}=\dfrac{|M(G)\cap I_\alpha(G)|-|M(G)\cap\overline{I_\alpha}(G)|}{|I_\alpha(G)|}.\\
Q_{\alpha,\Delta}'(M_p)&=\lim\limits_{n\to\infty}\;\min_{G\in\cl{G}(n,\Delta)}\cl{C}'(M_p(G),I_\alpha(G)).\\
\cl{C}''(M(G),I_\alpha(G))&=\begin{cases}
\dfrac{|M(G)\cap I_\alpha(G)|}{2|I_\alpha(G)|}+\dfrac{|\overline{M}(G)\cap \overline{I_\alpha}(G)|}{2|\overline{I_\alpha}(G)|},&\overline{I_\alpha}(G)\neq\emptyset\\
\dfrac{|M(G)|}{2n}+\dfrac{1}{2},&\overline{I_\alpha}(G)=\emptyset.
\end{cases}\\
Q_{\alpha,\Delta}''(M_p)&=\lim\limits_{n\to\infty}\;\min_{G\in\cl{G}(n,\Delta)}\cl{C}''(M_p(G),I_\alpha(G)).
\end{flalign*}
In this appendix we prove upper and lower bounds on $ Q_{\alpha,\Delta}' $ and $ Q_{\alpha,\Delta}'' $ which lead to the following graphs parallel to the graphs in Figure~\ref{fig: main results}. Just as with Figure~\ref{fig: main results}~(b), the mechanisms which give the best qualities in Figure~\ref{fig: appx results 1}~(b) and Figure~\ref{fig: appx results 2}~(b) are trivial.

\begin{figure}[h]
		\centering
	\begin{tabular}{l|l}
		\begin{tikzpicture}
		\begin{axis}[
		width=0.3\textwidth,
		height=0.25\textwidth,
		axis lines=middle,
		xlabel=$\alpha$,
		ylabel=$Q'$,
		enlargelimits,
		ytick={1},
		yticklabels={1},
		xtick={1},
		xticklabels={1}]
		
		\addplot[name path=F,line width=0.5mm,red,domain={0:1}] {(1+3*x)/(2+2*x)};
		
		\addplot[name path=G,line width=0.5mm,blue,domain={0:1}] {x};
		
		\addplot[pattern=north west lines, pattern color=green!50]fill between[of=F and G, soft clip={domain=0:1}];
		\addplot[name path=H,dashed,gray,domain={0:1}] {1};		
		
		\end{axis}
		\node at (50pt,-20pt) {(a)};
		\end{tikzpicture} &
		\begin{tikzpicture}
		\begin{axis}[
		width=0.3\textwidth,
		height=0.25\textwidth,		
		axis lines=middle,
		xlabel=$\alpha$,
		ylabel=$Q'$,
		enlargelimits,
		ytick={1},
		yticklabels={1},
		xtick={0.5,1},
		xticklabels={$ \tfrac{1}{2} $,1}]
		
		\addplot[name path=F,line width=0.5mm,black!60!green,domain={0:1/2}] {0};
		
		\addplot[name path=G,line width=0.5mm,black!60!green,domain={1/2:1}] {2-1/x};
		\end{axis}
		\node at (50pt,-20pt) {(b)};		
		\end{tikzpicture} 		
	\end{tabular}
	\caption{Our bounds for $ Q_{\alpha,\Delta}' $ as a function of $ \alpha $ when (a) $ \Delta=o(n) $ and (b) $ \Delta=n-1 $.}
	\label{fig: appx results 1}
\end{figure}

\begin{figure}[h]
	\centering
	\begin{tabular}{l|l}
		\begin{tikzpicture}
		\begin{axis}[
		width=0.3\textwidth,
		height=0.25\textwidth,
		axis lines=middle,
		xlabel=$\alpha$,
		ylabel=$Q''$,
		enlargelimits,
		ytick={0, 0.5,0.875, 1},
		yticklabels={o, $ \tfrac{1}{2} $,$ \tfrac{7}{8} $,1},
		xtick={1},
		xticklabels={1},
		xmin=0,
		xmax=1,
		ymin=0,
		ymax=1
		]
		
		\addplot[name path=F,line width=0.5mm,red,domain={0:1}] {3/4+1/(4*(2-x))};
		
		\addplot[name path=G,line width=0.5mm,blue,domain={0:1}] {(1+x)/2};
		
		\addplot[pattern=north west lines, pattern color=green!50]fill between[of=F and G, soft clip={domain=0:1}];
		\addplot[name path=H,dashed,gray,domain={0:1}] {1};		
		
		\end{axis}
		\node at (50pt,-20pt) {(a)};
		\end{tikzpicture} &
		\begin{tikzpicture}
		\begin{axis}[
		width=0.3\textwidth,
		height=0.25\textwidth,		
		axis lines=middle,
		xlabel=$\alpha$,
		ylabel=$Q''$,
		enlargelimits,
		ytick={0.5,1},
		yticklabels={$ \tfrac{1}{2} $,1},
		xtick={1},
		xticklabels={1}
		]
		
		\addplot[name path=F1,line width=0.5mm,black!60!green,domain={0:1}] {1/2};
		
		
		\end{axis}
		\node at (50pt,-20pt) {(b)};		
		\end{tikzpicture} 		
	\end{tabular}
	\caption{Our bounds for $ Q_{\alpha,\Delta}'' $ as a function of $ \alpha $ when (a) $ \Delta=o(n) $ and (b) $ \Delta=n-1 $.}
	\label{fig: appx results 2}
\end{figure}

\subsection{Bounds on $ Q_{\alpha,\Delta}' $}
To prove the lower bound of  we define two trivial schemes: $ M_N(G)\equiv N $ and $ M_\emptyset(G)\equiv\emptyset $. For any graph $ G $,
\begin{flalign*}
\cl{C}'(M_N(G), I_\alpha(G))&=\dfrac{|N\cap I_\alpha(G)|-|N\cap\overline{I_\alpha}(G)|}{|I_\alpha(G)|}=1-\dfrac{|\overline{I_\alpha}(G)|}{|I_\alpha(G)|}\geq 1-\dfrac{1-\alpha}{\alpha}=2-\dfrac{1}{\alpha}.\\
\cl{C}'(M_\emptyset(G),I_\alpha(G))&=\dfrac{|\emptyset\cap I_\alpha(G)|-|\emptyset\cap\overline{I_\alpha}(G)|}{|I_\alpha(G)|}=0.
\end{flalign*}
Thus when $ \alpha\leq \dfrac{1}{2} $ we have a trivial mechanism with quality 0 and when $ \alpha\geq \dfrac{1}{2} $ we have a trivial mechanism with quality $ 2-\dfrac{1}{\alpha} $, which is exactly the graphs in Figure~\ref{fig: appx results 1}~(b). The next proposition proves that when $ \Delta=n-1 $ there are no better mechanisms.
\begin{prop}\label{prp: apx1 high delta}
	Let $ m=\min\left \{2(1-\alpha), \dfrac{\Delta}{n} \right \} $. Then $ Q_{\alpha,\Delta}'\leq \max\left \{1-\dfrac{m}{2\alpha},0 \right \}$.
\end{prop}
To see how this proposition implies the upper bound of Figure~\ref{fig: appx results 1}~(b), take $ \Delta=n-1 $ and notice that \[ m=\begin{cases}
1,&\alpha\leq\dfrac{1}{2}\\
2(1-\alpha),&\alpha\geq\dfrac{1}{2}.
\end{cases} \]
Hence, when $ \alpha<\dfrac{1}{2} $, $ 1-\dfrac{m}{2\alpha}=1-\dfrac{1}{2\alpha}<0  $ and $ Q_{\alpha,\Delta}'\leq 0 $. When $ \alpha\geq\dfrac{1}{2} $, $ 1-\dfrac{m}{2\alpha}=1-\dfrac{1-\alpha}{\alpha}=2-\dfrac{1}{\alpha}\geq 0 $, and $ Q'_{\alpha,\Delta}\leq 2-\dfrac{1}{\alpha} $.
\begin{proof}[Proof of Proposition~\ref{prp: apx1 high delta}]
	We analyse the same series of networks as in Proposition~\ref{prp: large delta}. We get for Case I ($ m\geq 2\alpha $):
	\begin{flalign*}
	\cl{C}'(M(G),I_\alpha(G))\leq \dfrac{\mu|A|-\mu|B|}{|A|}=0.
	\end{flalign*}
	For Case II ($ m<2\alpha $):
	\begin{flalign*}
	\cl{C}'(M(G),I_\alpha(G))\leq \dfrac{\mu|A|+|C|-\mu|B|}{|A|+|C|}=\dfrac{|C|}{|A|+|C|}=\dfrac{\alpha-m/2}{\alpha }=1-\dfrac{m}{2\alpha}.
	\end{flalign*}
	Combining the two cases, we get the claimed bound on $ Q_{\alpha,\Delta}'$.
\end{proof}
The next proposition proved the upper bound in Figure~\ref{fig: appx results 1}~(a).
\begin{prop}
	For all $ \alpha,\Delta $, $ Q_{\alpha,\Delta}'\leq \dfrac{1+3\alpha}{2+2\alpha}$.
\end{prop}
\begin{proof}
	We use the same network structure of Proposition~\ref{prp: small delta} and get:
	\begin{flalign*}
	\cl{C}'(M(G),I_\alpha(G))&\leq \dfrac{|C|+1-\mu|A|}{|C|+1}=1-\dfrac{\mu|A|}{\alpha n}.\\
	\cl{C}'(M(G'),I_\alpha(G'))&\leq \dfrac{|A|+\mu|B|+|C|+1}{|N|}=1-\dfrac{(1-\mu)|B|}{n}.
	\end{flalign*}
	Comparing the two expression to get the highest minimum of the two:
	\begin{flalign*}
	\cl{C}'(M(G),I_\alpha(G))=\cl{C}'(M(G'),I_\alpha(G'))\iff\\
	\dfrac{\mu|A|}{\alpha n}=\dfrac{(1-\mu)|B|}{n}\iff\\
	\dfrac{\mu}{\alpha}=1-\mu\iff \mu=\dfrac{\alpha}{1+\alpha}.
	\end{flalign*}
	Hence, 
	\begin{flalign*}
	Q_{\alpha,\Delta}'\leq 1-\dfrac{|A|}{(1+\alpha)n}=1-\dfrac{1-\alpha}{2(1+\alpha)}=\dfrac{1+3\alpha}{2+2\alpha}.
	\end{flalign*}
\end{proof}
The next proposition in the lower bound of Figure~\ref{fig: appx results 1}~(a).
\begin{prop}
	For any $ \alpha $ and any $ \Delta\leq\alpha n $, there is an IC probabilistic classification mechanism with quality $ \alpha-\dfrac{\Delta}{n} $.
\end{prop}
\begin{proof}
	Let $ \beta=\alpha-\Delta/n $ and let $ W(G) $ be the set of absolutely worthy agents as defined in the proof of Proposition~\ref{prp: probabilistic mechanism}. We define mechanism $ M $:
	\begin{flalign*}
	\forall G\in\cl{G}(n,\Delta)\; M(G)=W(G).
	\end{flalign*}
	Then,
	\begin{flalign*}
	\cl{C}'(M(G),I_\alpha(G))=\dfrac{|W(G)|}{|I_\alpha(G)|}\geq \dfrac{\beta n}{n}=\beta.
	\end{flalign*}
\end{proof}

\subsection{Bounds on $ Q_{\alpha,\Delta}'' $}
As a trivial mechanism we take $ M_N(G)\equiv N $. Clearly, $ \cl{C}''(M_N(G),I_\alpha(G))=\dfrac{1}{2} $ for all networks. 
\begin{prop}
	For $ \Delta=n-1 $ and any $ \alpha $, $ Q_{\alpha,\Delta}''\leq \dfrac{1}{2} $.
\end{prop}
\begin{proof}
	Partition $ N $ into two sets $ A, B $ of sizes $ \alpha n $ and $ (1-\alpha)n $, respectively. Let $ G $ be the complete network in which the weights on the out-edges of nodes in $ A $ are 0, and the weights on the out-edge of nodes in $ B $ are 1. Clearly, the nodes in $ A $ are symmetric and the nodes in $ B $ are symmetric. The nodes in $ A $ have a better score than the nodes in $ B $, hence $ I_\alpha(G)=A$. Let $ M $ be a symmetric probabilistic mechanism. By similar argument to that in the proof of Proposition~\ref{prp: large delta}, we may assume that all the nodes in the network are selected by $ M $ with the same probability, $ \mu $. We get,
	\begin{flalign*}
	\cl{C}'(M(G),I_\alpha(G))=\dfrac{\mu|A|}{2|A|}+\dfrac{(1-\mu)|B|}{2|B|}=\dfrac{1}{2}.
	\end{flalign*}
\end{proof}
This settles Figure~\ref{fig: appx results 2}~(b). We turn to the upper bound of Figure~\ref{fig: appx results 2}~(a).

\begin{prop}
	For all $ \alpha,\Delta $, $ Q_{\alpha,\Delta}'\leq \dfrac{3}{4}+\dfrac{1}{4(2-\alpha)}$.
\end{prop}
\begin{proof}
	We use the same network structure of Proposition~\ref{prp: small delta} and get:
	\begin{flalign*}
	\cl{C}'(M(G),I_\alpha(G))&\leq \dfrac{|C|+1}{2(|C|+1)}+\dfrac{(1-\mu)|A|+|B|}{2(|A|+|B|)}=1-\dfrac{\mu|A|}{2(|A|+|B|)}=1-\dfrac{\mu}{4}.\\
	\cl{C}'(M(G'),I_\alpha(G'))&\leq \dfrac{|A|+\mu|B|+|C|+1}{2|N|}+\dfrac{1}{2}=1-\dfrac{(1-\mu)|B|}{2n}=1-\dfrac{(1-\mu)(1-\alpha)}{4}.
	\end{flalign*}
	Comparing the two expression to get the highest minimum of the two:
	\begin{flalign*}
	\cl{C}'(M(G),I_\alpha(G))=\cl{C}'(M(G'),I_\alpha(G'))\iff\\
	\mu=(1-\mu)(1-\alpha)\iff \mu=\dfrac{1-\alpha}{2-\alpha}.
	\end{flalign*}
	Hence, 
	\begin{flalign*}
	Q_{\alpha,\Delta}'\leq 1-\dfrac{1-\alpha}{4(2-\alpha)}=\dfrac{3}{4}+\dfrac{1}{4(2-\alpha)}.
	\end{flalign*}
\end{proof}

The last proposition is the mechanism which gives the lower bound of Figure~\ref{fig: appx results 2}~(a).
\begin{prop}
	For any $ \alpha $ and any $ \Delta\leq\alpha n $, there is an IC probabilistic classification mechanism with quality $ \dfrac{1}{2}\left(1+\alpha-\dfrac{\Delta}{n} \right) $.
\end{prop}
\begin{proof}
	Let $ \beta=\alpha-\Delta/n $ and let $ W(G) $ be the set of absolutely worthy agents as defined in the proof of Proposition~\ref{prp: probabilistic mechanism}. We define mechanism $ M $:
	\begin{flalign*}
	\forall G\in\cl{G}(n,\Delta)\; M(G)=W(G).
	\end{flalign*}
	Then,
	\begin{flalign*}
	\cl{C}'(M(G),I_\alpha(G))\geq\dfrac{|W(G)|}{2|I_\alpha(G)|}+\dfrac{|B(G)\cap \overline{I_\alpha}(G)|+|U(G)||}{2|\overline{I_\alpha}(G)|}
	\geq \dfrac{\beta n}{2n}+\dfrac{1}{2}=\dfrac{1}{2}(1+\beta).
	\end{flalign*}
\end{proof}